\documentclass[letterpaper,10pt,journal,onecolumn]{IEEEtran}
\usepackage{graphicx}
\usepackage{gensymb}
\usepackage{amsmath, amssymb, amsthm}
\usepackage{pdfpages}
\usepackage{algorithm}
\usepackage[noend]{algpseudocode}
\usepackage{verbatim}
\usepackage{booktabs}
\usepackage{textcomp}
\newtheoremstyle{assumptionstyle}{0in}{0in}{\normalfont}{0.5em}{\itshape}{:}{.5em}{}
\newtheoremstyle{propertystyle}{0in}{0in}{\normalfont}{}{\bf}{:}{.5em}{}
\theoremstyle{plain}
\newtheorem{theorem}{Theorem}

\theoremstyle{propertystyle}
\newtheorem{assumption}{Assumption}
\theoremstyle{propertystyle}
\newtheorem{property}{Property}
\usepackage{setspace}
\usepackage{moresize}
\usepackage[linktocpage=true,hyperfootnotes=false]{hyperref}

\begin{document}
\setlength{\abovedisplayskip}{3pt}
\setlength{\belowdisplayskip}{3pt}

\title{Trajectory Generation for UAVs in Unknown Environments with Extreme Wind Disturbances}

\author{Kenan Cole,
\thanks{Kenan Cole is a graduate of the Mechanical and Aerospace Engineering Department,     The George Washington University, Washington, DC 20052, USA}
        ~Adam M. Wickenheiser%
\thanks{Adam M. Wickenheiser is an Associate Professor in the Mechanical Engineering Department, University of Delaware, Newark, DE 19716, USA}}
        
\maketitle

\begin{abstract}
The widespread use of unmanned aerial vehicles (UAVs) by the military, commercial companies, and academia continues to push research for autonomous vehicle navigation, particularly in varying environmental conditions and beyond-line-of-sight (BLOS) applications. This article addresses trajectory generation for UAVs operating in extreme environments where the wind disturbances may exceed the vehicle's closed-loop stability bounds. To do this, a controller is developed that has two modes of operation: (1) normal mode, and (2) drift mode. In the normal mode the vehicle's thrust and sensor limitations are not exceeded by environmental conditions, whereas in the drift mode they are. In the drift mode, a drift frame that moves with the prevailing wind is established in which the vehicle maintains control authority to generate and track trajectories. The vehicle maintains control authority by relaxing the inertial frame trajectory tracking requirement and re-planning the trajectory in the drift frame. Guarantees are established to ensure tracking of the trajectory, collision avoidance, and respecting the vehicle thrust and sensor limitations. Simulation results demonstrate the algorithm properties through two scenarios. First, the performance of two quadrotors is compared where one utilizes the drift mode and the other does not. Second, multiple vehicles navigate through two narrow openings between protected and windy environments to demonstrate on-board updates to navigation parameters based on environmental conditions.
\end{abstract}

\section{Introduction}
Unmanned aerial vehicle use continues to grow with an increasing push for autonomous navigation and BLOS (also called beyond visual line of sight - BVLOS) operations. Current FAA guidelines allow certified remote pilots to fly drones according to the FAA's Part 107 rules and require waivers for any flight outside the rules \cite{FAACertifiedPilots}. As of June 2019 there were 38 waivers for Part 107.31 which is for flying unmanned aircraft systems (UAS) beyond the visual sight of the pilot \cite{FAA107Waivers}. Examples of waivers requested include agricultural and land management monitoring (PrecisionHawk, GreenSight Agronomics), infrastructure inspection (Xcel energy, BNSF Railway), transportation and insurance (Kansas Department of Transportation Division of Aviation, State Farm), and UAV development (CyPhy, Airobotics, Project Wing).

In all of these examples, BLOS operation includes several challenges, particularly the environment. It cannot be assumed that the vehicle has a priori knowledge of the environment and/or that the environment is static. Furthermore, there may be obstacles and/or other vehicles in the environment requiring the vehicle to sense and avoid safety threats using only on-board sensors. 

These uncertainties preclude global optimization algorithms \cite{VanLoock2014}, which require a priori knowledge of the environment, and local planners \cite{Shiller2013} are unattractive given their lack of convergence guarantees. Reactive trajectory generators, \cite{Matveev2015,Ferrera2017}, provide a strong basis for trajectory generation in unknown environments by generating trajectories directly as the environment is sensed. Their drawback for this application is they do not utilize the wind disturbance as an input to the trajectory generation. 

In order to guarantee stability of the vehicle, the trajectory must be generated while taking into account the control authority required to overcome any anticipated disturbances. If it does not, then under extreme wind conditions the trajectory could be too aggressive for the vehicle and violate the vehicle controller stability requirements. Hardware constraints such as sensor range, maximum velocity, clearance radius, and/or turning radius have been considered \cite{Matveev2015}-\nocite{Ferrera2017}\nocite{Hoy2012}\nocite{Choi2013}\cite{Chunyu2010}, but in each case the trajectory generation is decoupled from the disturbance rejection.

In our previous work \cite{Cole2018} we addressed this shortcoming by presenting a reactive trajectory generation algorithm for vehicles with second-order dynamics that considered the vehicle hardware constraints but also included the disturbance as an input. This work utilized an a priori worst case wind estimate to set trajectory generator parameters, but in a BLOS environment a predicted maximum wind cannot be guaranteed. Additionally, setting the parameters a priori results in conservative behavior when the conditions are not worst-case and results in grounding the vehicle when the wind exceeds the closed-loop stability.  

There are relatively few studies that consider extreme operation of UAVs. In one example, Dicker \cite{Dicker2016} develops a controller for a vehicle to recover from a collision by considering altitude control and horizontal control separately. Faessler et al.~\cite{Faessler2015} develop a controller for a vehicle that is being thrown and then must stabilize. Altitude is again the primary focus of the controller. In both of these cases, the extreme condition of a collision or being thrown is not persistent, and the vehicle has the control authority to re-stabilize itself. This differs from the conditions of this article where the vehicles may need to sustain safe operation in persistent, extreme environments by stabilizing about a new reference state.

To address the need for sustained operation in extreme environments that may violate the closed-loop stability bounds, we build upon our previous work \cite{Cole2018} by extending the environments in which the trajectory generation algorithm is valid. To do this, the algorithm operates in two modes, where the first mode is for normal operation when the vehicle's thrust and sensor limitations are not exceeded. The second mode is a drift mode that extends the vehicles' stability into higher wind cases. In drift mode a drift frame is established, where within the drift frame the vehicle has control authority to generate and track trajectories. Stability in the drift frame is possible by relaxing the inertial frame trajectory and instead re-planning the trajectory in the drift frame. The vehicle uses an estimate of the wind to establish a viable drift frame. 

The rest of the article is organized as follows. First, Sec.~\ref{SecMotivate} provides a motivating example for the work. Next, Sec.~\ref{SecProbDefCh6} defines the problem and establishes assumptions for the vehicles and operating environment. The trajectory generation algorithm from our previous work is briefly summarized in Sec.~\ref{SecTrajGen}, describing how the vehicle makes smooth heading and velocity changes. Section \ref{SecDriftFrame} develops the parameter dependencies on the environmental wind conditions as well as the solution for the drift frame velocity. Next, Sec.~\ref{SecVehController} defines the vehicle dynamics and flight controller. Simulation scenarios demonstrate the algorithm properties in Sec.~\ref{SecSimScenarioFormation}. Finally, Sec.~\ref{SecConclusionFormation} provides concluding remarks.

\section{Motivating Example}
\label{SecMotivate}
Consider an example where two identical vehicles experience identical gusts that are strong enough to exceed their closed-loop stability bounds. One vehicle does not enable the drift mode, tries to continue tracking the inertial trajectory, and becomes unstable, eventually crashing. The other vehicle enables the drift mode, tracks a new trajectory in the drift frame (thus relaxing the trajectory in the inertial frame), and maintains control. These different behaviors are shown in Fig.~\ref{FigMotivatingExample}.

\begin{figure*}
	\begin{center}
		\includegraphics[width=0.9\textwidth]{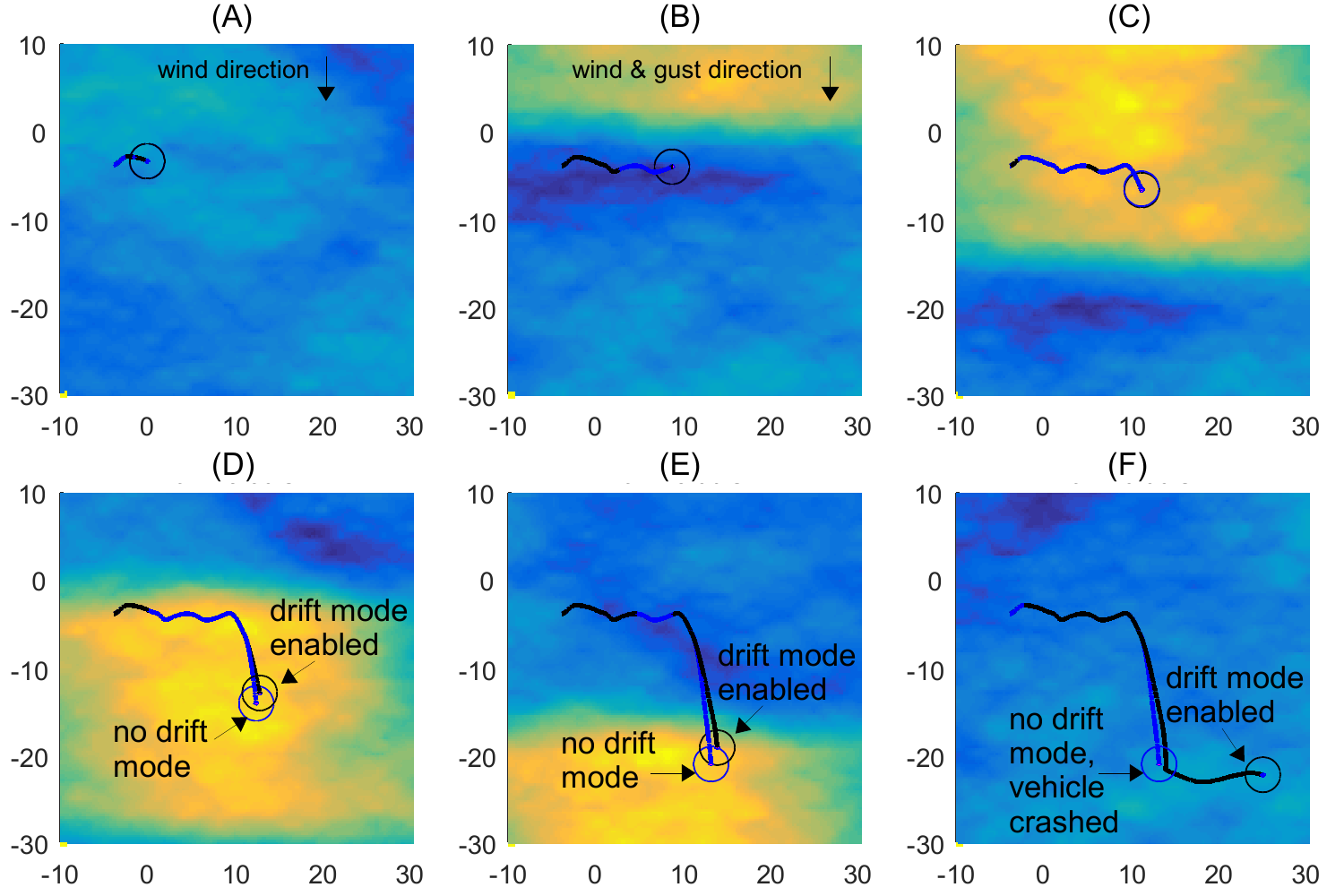}
		\caption{\label{FigMotivatingExample} Example showing two vehicles experiencing a large gust. The vehicle trajectories are overlaid to demonstrate the difference in performance. (A) Both vehicles start traversing the environment tracking inertial trajectories. (B) The vehicles continue to track inertial trajectories as a gust approaches (yellow region). (C) The gust is stronger than the vehicles can combat, so both vehicles are pushed by the gust. (D) The trajectories of the vehicles start to diverge, where the vehicle enabling drift mode has more control than the one that does not. (E) The vehicle that does not enable drift mode has tried to combat the gust, lost altitude, and crashed. The vehicle utilizing drift mode is still stable and tracking a trajectory in the drift frame. (F) Now that the gust has passed, the vehicle that enabled drift mode starts to correct back to the desired inertial frame trajectory.}
	\end{center}
\end{figure*}

This example illustrates the utility of the drift mode to enable vehicles to maintain stability and expand the conditions in which they can fly despite sustained or temporary gusting events. In the event of a temporary gust the vehicle disables drift mode and goes back to its normal operation when the gust is over. Our previous work \cite{Cole2018} only addressed inertial frame trajectories and would result in instability for a gust that exceeded the set maximum.

This example also highlights the need for the controller to have knowledge of the wind conditions. If the controller is ignorant of the current wind, then even in light or moderate winds the controller may demand greater thrust than the vehicle can provide, thus resulting in similar behavior to the vehicle that does not enable drift mode.   

\section{Problem Definition}
\label{SecProbDefCh6}
A trajectory generation algorithm is defined that establishes two operational modes for vehicles that allows the vehicles to extend their operation into environments that may their violate closed-loop stability bounds. The algorithm satisfies Properties \ref{PropTrajGen1} and \ref{PropTrajGen2} from our previous work \cite{Cole2018} in the inertial frame and Property \ref{PropDriftCh6} for the drift frame established in this article. These properties are rigorously achieved under the assumptions that follow.

\subsection{Algorithm Properties}
\label{SubSecAlgPropsCh6}
{\color{black}\begin{property} \label{PropTrajGen1} Generation of a piecewise-smooth (with isolated bounded discontinuities) desired trajectory $\mathbf{p}_d \in \mathbb{R}^3$ where the derivatives $\mathbf{p}_d^{(i)} \in \mathbb{R}^3,~\forall i=0,1,\ldots,n$ exist, are bounded, and respect the vehicle's maximum thrust, $f_{max}$, for a translational wind velocity of unknown direction and bounded magnitude, $||\mathbf{v}_{air}|| \leq v_{air,max}$.
\end{property}
\begin{property} \label{PropTrajGen2} Clearance of all obstacles and other vehicles by a user-defined clearance radius, $r_c$, which takes into account the vehicle's size as well as measurement, estimation, and tracking errors. 
\end{property}}
\begin{property} \label{PropDriftCh6} Generation of trajectories that satisfy Properties \ref{PropTrajGen1} and \ref{PropTrajGen2} in a drift frame, $[\mathbf{x}_D,\mathbf{y}_D,\mathbf{z}_D] \in \mathbb{R}^{3\times3}$, that moves at constant drift velocity, $\mathbf{v}_{drift}$, for conditions in the inertial frame where the drag force from the maximum translational wind velocity, $K_d v_{air,max}^2$, exceeds the vehicle's maximum remaining in-plane thrust, $f_{planar}$, where $K_d = 1/2\rho C_d A_{x_w}$, $\rho$ is the air density, $C_d$ is the drag coefficient, and $A_{x_w}$ is the cross sectional area normal to the resultant wind vector. 
\end{property}

\subsection{Assumptions}
{\color{black}\label{SecAlgAssumpFormation}
	\begin{assumption}
	\label{AssumptionPlanar}
	Vehicle desired trajectories and obstacle motions are planar, but vehicle dynamics are not restricted to be planar.
	\end{assumption}
	\begin{assumption}
	\label{AssumpVehicles}
	 Vehicles are finite in number and heterogeneous in physical parameters (mass, max thrust, etc.) and importance (e.g. higher valued asset).  
	\end{assumption} 
	\begin{assumption}
	 \label{AssumpVehCommsInfo} Vehicles share current position and course information when in range via wireless communication.
	\end{assumption} 
	\begin{assumption}
	 \label{AssumpSensorComms} Vehicles sensor and communication sample periods and ranges are equal and given by $\Delta T_s$ and $r_s > r_c$, respectively. Within these limitations, the sensor and inter-vehicle communications provide perfect distance and velocity information. 
	\end{assumption}
	\begin{assumption}
	  The clearance radius $r_c$ ensures there are no aerodynamic interactions between one vehicle and another or with obstacles. 
	\end{assumption}
	\begin{assumption}
	 Wind disturbances are bounded, time-varying, and planar. Updraft effects near obstacles are assumed to be limited to a distance less than $r_c$.
	\end{assumption}
	\begin{assumption}
	\label{AssumpLessCapable}
	There are a finite number of obstacles and each obstacle is finite size, moves with constant velocity (less than minimum vehicle cruise velocity) and constant course. Minimum obstacle separation does not prevent the vehicles from moving between them.
	\end{assumption}
	\begin{assumption}
	\label{AssumpGoalPos}
	 Goal positions are not too close to obstacles or each other to violate vehicle clearance radii and are not infinitely far from the coordinate origin.
	\end{assumption}}
	\begin{assumption}
\label{AssumpMovingObsFormation}
The maximum wind in which moving obstacles can operate, $v_{air,max,obs}$, is always less than or equal to the maximum wind in which the vehicles can operate: $v_{air,max,obs} \leq v_{air,max}$.
	\end{assumption}
	
\section{Trajectory Generation}
\label{SecTrajGen}
{\color{black}The trajectory generation algorithm takes each vehicle from its current position and velocity and guides it on a collision-free trajectory to the goal position.} This section summarizes the relevant sections from our previous work \cite{Cole2018} and indicates any modifications for application to the drift frame development. The following steps of the algorithm are briefly summarized: (1) vehicle maneuverability ranking to determine which vehicles are responsible for maneuvering (Sec.~\ref{SubSecVehManeuver}), (2) vehicle sensors and how the vehicle compiles sensor input (Sec.~\ref{SubSecObsVehIdentify}), (3) procedure for determining a course change around an obstacle (Sec.~\ref{SecCourseChangeOneObs}), and (4) governing algorithms to make smooth course and velocity transitions (Sec.~\ref{SecSigmoid}). 

\subsection{Ranking Vehicles' Maneuverability}
\label{SubSecVehManeuver}
{\color{black}The vehicle utilizes on-board sensing and communication to identify obstacles and other vehicles within sensor range. When two vehicles meet, to determine which one is responsible for maneuvering, the vehicles exchange their maximum cruise velocity, $v_c$, current velocity, $\mathbf{\dot{p}}_d$, clearance radius, $r_c$, {\color{black}maximum wind speed in which they can operate, $v_{w,op}$}, and a pre-assigned $ID$ value when they come within communication range of each other.} 

{\color{black}The maximum wind speed is an additional parameter to exchange that was not included in our previous work. The development of the drift frame requires this parameter to be exchanged; otherwise a vehicle may be expected to maneuver when it does not have sufficient control authority, possibly resulting in a collision.} The cruise velocity is modified from \cite{Cole2018} for the ranking determination to account for vehicles that have no control authority in the environment, resulting in
\begin{equation}
	v_c^* = \left\{\begin{array}{ll}
				0, & v_{air} > v_{w,op} \\
				0, & ||\mathbf{\dot{p}}_d(t)|| = 0 \\
				v_c, & \mathrm{otherwise}
				\end{array}
				\right.
\end{equation} 
{\color{black}To satisfy Assumption \ref{AssumpLessCapable}, vehicles with larger $v_c^*$ maneuver around vehicles with smaller $v_c^*$. Likewise, vehicles with higher $ID$ values maneuver around vehicles with lower $ID$ values for equal $v_c^*$, forming the set $\mathcal{I}_{mnvr} \subseteq \mathcal{I}_{nr}$, where $\mathcal{I}_{nr}$ is the set of all vehicles within $r_s$ of the vehicle's current position.}

\subsection{Compiling Sensor Inputs}
\label{SubSecObsVehIdentify}
{\color{black}The vehicle uses on-board distance and angle measurements to obstacles and other vehicles to determine the most imminent collisions, if any. We assume that the sensing is isotropic (i.e. has the same range and rate in all directions) and that the vehicle differentiates obstacles by finding discontinuities in range and angle. The vehicle gives each distinct obstacle a unique local identifier, $id$, and compiles all obstacles in the set $\mathcal{I}_{obs} = \left\{id_1,...,id_m \right\}$, where $m$ is the number of distinct obstacles within range. {\color{black}The inertial positions of the sensed points are given by $\mathbf{p}_{id,i}$, where $i = 1,...,n_{id}$, and $n_{id}$ is the number of sensed points for that particular obstacle.}

The data for the vehicles in $\mathcal{I}_{mnvr}$ is combined with the data for the obstacles in $\mathcal{I}_{obs}$ to form a data array of distinct vehicles and obstacles in the environment.}

\subsection{Course Change Definition for an Obstacle}
\label{SecCourseChangeOneObs}
{\color{black}To safely navigate the environment and avoid collisions, the vehicle can change course, $\Delta \phi$, and/or velocity, $\Delta v$.} {\color{black} This section summarizes course changes; velocity changes are discussed in \cite{Cole2018}. 

The vehicle's overall objective is to reach the goal position; therefore, heading changes are only made when necessary and to minimize the time to reach the goal position. The course change definition provides a desired heading change, $\Delta \phi$, corresponding circumnavigation direction (i.e. clockwise or counterclockwise traversing of the obstacle), $\mathbf{z}_{k}$, and feasible set, $\mathbf{O}_k$, of all safe course changes to clear the obstacle by $r_c$. Figure \ref{FigTrajGenSummary} summarizes the algorithm process through an example in the inertial frame where the drift mode is not required. Section \ref{SecDriftFrame} expands the course change for application to the drift frame. }

\begin{figure*}
	\begin{center}
		\includegraphics[width=0.9\textwidth]{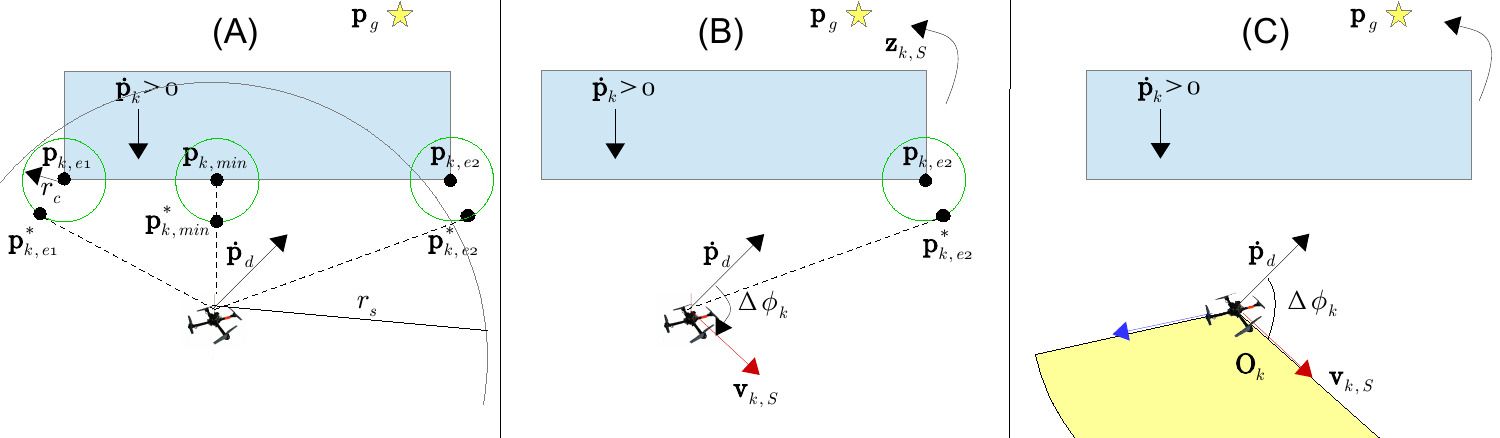}
		\caption{\label{FigTrajGenSummary} {\color{black}Summary of course change definition in response to an obstacle moving toward the vehicle. (A) Determination of the constraining geometry from the sensor input. The bounding extent points, $\mathbf{p}_{k,e1}$, $\mathbf{p}_{k,e2}$, corresponding projected extent points $\mathbf{p}^*_{k,e1}$, $\mathbf{p}^*_{k,e2}$, minimum sensed point, $\mathbf{p}_{k,min}$, and projected minimum point, $\mathbf{p}^*_{k,min}$, are all used to determine an appropriate course change and circumnavigation direction. (B) Choose extent point $\mathbf{p}_{k,e2}$ to traverse the obstacle towards based on sensor input, goal position location, and estimated traverse time. The heading change, $\Delta \phi_k$, is established from the course change vector, $\mathbf{v}_{k,S}$, as is a counter-clockwise circumnavigation direction, $\mathbf{z}_{k,S}$. (C) Determination of the feasible set of all course change angles, $\mathbf{O}_k$, (highlighted in yellow) the vehicle could make to clear the obstacle.} }
	\end{center}
\end{figure*}

\subsection{Smooth course and velocity transitions}
\label{SecSigmoid}
{\color{black}The trajectory generation algorithm utilizes sigmoid functions to transition from the previous course, $\phi_{n-1}$, and velocity, $v_{n-1}$, to a new course, $\phi_{n}$, and velocity, $v_{n}$. The sigmoid functions are
\begin{align}
			\label{EqSigPhi} \phi &= c_1 \tanh(c_2 \tau - c_3) + c_4 \\
			\label{EqSigV} v &= d_1 \tanh(d_2 \tau - d_3) + d_4
\end{align}
\noindent where the coefficients are chosen to match thrust, sensor, and $r_c$ constraints. 

As new information is provided, additional sigmoids are added to the current sigmoid to still respect the thrust constraints of the vehicle. This is done by an offset time to start the new sigmoid and solving for the new sigmoid curve timespan. Both of these parameters take into account the maximum thrust, sensor update rate and range constraints, and clearance radius. The transitions between the sigmoid curves are smooth and guarantee the vehicle can track the desired trajectory. }

\section{Drift Frame Development}
\label{SecDriftFrame}
The drift frame utilizes the trajectory generation algorithm to extend the environments in which the vehicle can operate and provide quantitative bounds for that environment. To do this, there are three extensions to the trajectory generation algorithm that are developed in this section. First, the vehicle must make a wind estimation, $\tilde{\mathbf{v}}_{air}$, to update the parameters for the drift frame, next the vehicle updates the parameters which include $r_c$ and $v_c$ (Sec.~\ref{SubSecParamDependOnVair}), and finally the vehicle determines the drift frame velocity, $\mathbf{v}_{drift}$, (Sec.~\ref{SubSecDriftFrameVel}). 

\subsection{Wind Estimation}
\label{SubSecWindEst}
To determine the drift frame velocity, the vehicle must make a wind estimation. The wind estimation can be achieved through several different approaches such as inverse dynamics \cite{Waslander2009,Sikkel2016}, linear fit \cite{Palomaki2017}, or use of additional sensors \cite{Sydney2013,Bruschi2016,Kumar2016}. 

It is assumed in this work that any additional capacity of the vehicle is taken up by a mission payload leaving no room for additional sensors. For the purposes of this article we use inverse dynamics with a nonlinear solver to estimate the wind, but this is not the only suitable wind estimation method. 

\subsection{Parameter Dependency on Wind Estimation}
\label{SubSecParamDependOnVair}
Our previous work used a maximum gust value to calculate all parameters associated with the trajectory, which resulted in conservative behavior. As we allow for larger gusting events to be handled, we also allow the clearance radius, $r_c$, of the vehicle (i.e. the distance by which it clears other vehicles and obstacles) and the cruise velocity (i.e. maximum safe velocity to navigate the environment) to change as functions of the environmental conditions. This allows the vehicle to maneuver more or less aggressively as conditions change.

First we consider the update to the clearance radius. There is always a minimum clearance radius associated with the vehicle size, $r_{c,v}$, but the additional clearance to take into account controller tracking errors and environmental effects, $r_{c,e}$, is updated with the changing environment. Allowing the clearance radius, $r_c = r_{c,v} + r_{c,e}$, to change reduces the required maneuvering in lighter wind conditions, allowing the vehicle to take a more optimal course. 

Once the clearance radius is updated, the cruise velocity is re-solved according to Theorem 2 in \cite{Cole2018}. The clearance radius and cruise velocity are inversely proportional so a larger clearance radius reduces the cruise velocity. Both parameters are then used to update the minimum reaction distance, which is the distance at which the vehicle must start maneuvering to avoid a collision. Likewise, the parameters are used for computing course changes (Sec.~\ref{SecCourseChangeOneObs}) and smooth transitions (Sec.~\ref{SecSigmoid}).

As a result of changing $r_c$ values, there may be temporary violations of $r_c$ right after a new clearance radius is computed, as shown in Fig.~\ref{FigSuddenRCIncrease}. This temporary violation does not mean that the vehicle experiences a collision since the vehicle size is accounted for with $r_{c,v} < r_c$. Instead, the generated trajectory guides the vehicle out of the violation to continue to ensure safe navigation.

To do this, the computations that are performed to the sensor points to account for $r_c$, as shown in Fig.~\ref{FigTrajGenSummary}A, must be modified when there is a temporary $r_c$ violation. The updated points are defined according to Fig.~\ref{FigInRc} as follows:

\begin{equation}
	\label{EqPhiEiFormation}
	\phi_{e,i} = \left\{\begin{array}{ll}
					\phi_{e1,i} + \phi_{e2,i}, & ||\mathbf{r}_{k,i}|| > r_c \\
					\mathrm{angle}(\mathbf{r}_{k,min},\mathbf{p}_{k,i}^*), &||\mathbf{r}_{k,i}|| \leq r_c 
					\end{array}
					\right.
\end{equation}
\noindent where
\begin{align}	
	&\phi_{e1,i} = \mathrm{angle}(\mathbf{r}_{k,min},\mathbf{r}_{k,i}) \\
	\label{EqPhiE2PadObsFormation}
	&\phi_{e2,i} = k_{\phi,e1,i}\left(\phi_{e1,i}\right)\sin^{-1}\left(\frac{r_c}{||\mathbf{r}_{k,i}||}\right) \\
	&\mathbf{r}_{k,i} = \mathbf{p}_{k,i}-\mathbf{p}_d \\
	&\phi_{\pm1} =  \frac{\pi}{180}\mathrm{sgn}(\left(\mathrm{angle}(\mathbf{\dot{p}}_d,\mathbf{r}_{k,min})\right)) \\
	&k_{\phi,e1,i} = \left\{\begin{array}{ll}
	\mathrm{sgn}(\phi_{e1,i}), & |\phi_{e1,i}| > 0 \\
	\mathrm{sgn}(\mathrm{angle}(\mathbf{r}_{k,min},\mathbf{\dot{p}}_d)), & \phi_{e1,i} = 0
	\end{array}
	\right. \\
	\label{EqPkiStarFormation}
	&\mathbf{p}_{k,i}^* = \left\{\begin{array}{ll}
				\mathbf{p}_d + \sqrt{||\mathbf{r}_{k,i}||^2 - r_c^2}\mathbf{R}_{\phi_{e,i}}\hat{\mathbf{r}}_{k,min}, & ||\mathbf{r}_{k,i}|| > r_c \\
				\mathbf{p}_{k,i} + r_c \mathbf{R}_{\phi_{\pm1}}\frac{-\mathbf{r}_{k,i}}{||\mathbf{r}_{k,i}||}, & ||\mathbf{r}_{k,i}|| \leq r_c 
				\end{array}
				\right.
\end{align}

\begin{figure}
	\begin{center}
		\includegraphics[width=3in]{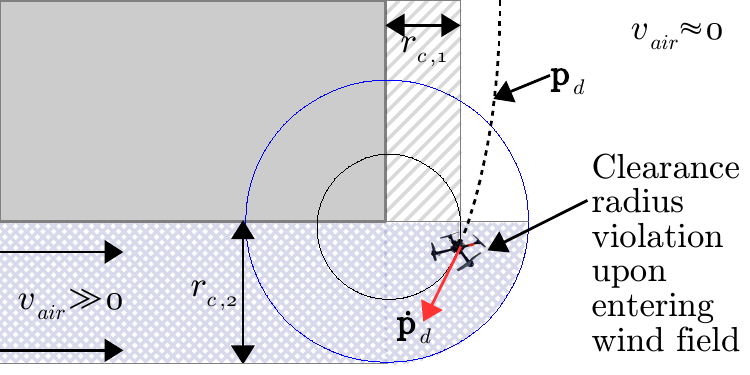}
		\caption[As the vehicle leaves the protected area and enters the wind field,]{\label{FigSuddenRCIncrease} The vehicle leaves a protected area and enters a wind field, which requires updating the clearance radius, $r_{c,1} \rightarrow r_{c,2}$. Since the vehicle tracks a trajectory around the obstacle for, $r_{c,1}$, there is a temporary clearance radius violation.}
	\end{center}
\end{figure}

\begin{figure}
	\begin{center}
		\includegraphics[width=3in]{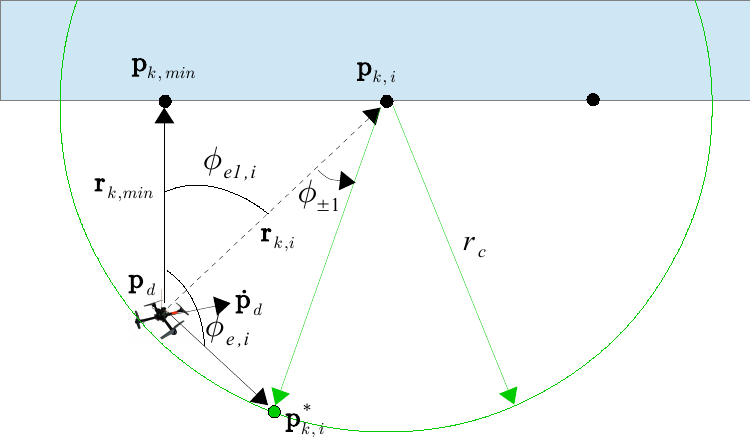}
		\caption{\label{FigInRc} Determination of candidate course change points for a vehicle that is temporarily violating $r_c$.}
	\end{center}
\end{figure}

Similarly, the tangent direction (i.e. a vector to traverse the obstacle either along the ``face" of the obstacle or along a more conservative path that opens space between the vehicle and obstacle) is also affected by the temporary $r_c$ violation. The previous tangent direction definitions from \cite{Cole2018} did not include temporary $r_c$ violations because $r_c$ was set for the maximum wind condition. If the original definitions are used, the vehicle will continue to violate $r_c$. Instead, the tangent direction definition is modified to safely navigate the vehicle away from the obstacle and out of the violation. This update is shown in Fig.~\ref{FigTangentDirectionsB} and defined in Eqs.~\ref{EqPljS1Veh} to \ref{EqPljS4Veh} which are repeated from \cite{Cole2018} for clarity (note these equations are the same as Eqs.~44 to 47 in \cite{Cole2018} which govern behavior for navigating around other vehicles):   

\begin{figure}
	\begin{center}
		\includegraphics[width=3in]{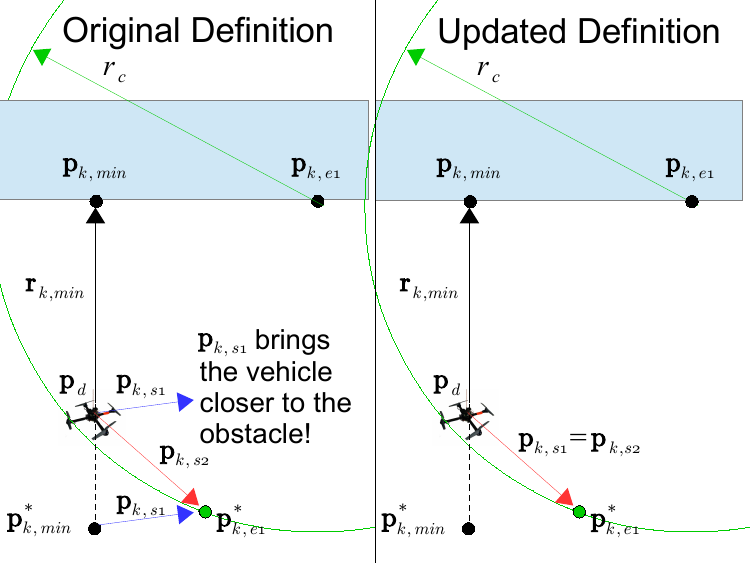}
		\caption{\label{FigTangentDirectionsB} Tangent direction definitions when there are temporary $r_c$ violations. In the original definition, the $p_{k,s1}$ tangent direction brings the vehicle closer to the obstacle, which is undesired. Instead the updated definition navigates the vehicle to clear the $r_c$ violation. }
	\end{center}
\end{figure}
{\color{black}
{\small
\begin{align}
	\label{EqPljS1Veh}
	\mathbf{p}_{k,s1} &= \left\{\begin{array}{ll}
						\mathbf{p}_{k,e1}^* - \mathbf{p}_{k,min}^*, & ||\mathbf{p}_{k}-\mathbf{p}_d|| \geq r_c \\
						\mathbf{p}_{k,e1}^* - \mathbf{p}_d, & ||\mathbf{p}_{k}-\mathbf{p}_d|| < r_c
						\end{array}
						\right. \\ 
	\mathbf{p}_{k,s2} &= \mathbf{p}_{k,e1}^* - \mathbf{p}_d \\
	\mathbf{p}_{k,s3} &= \left\{\begin{array}{ll}
						\mathbf{p}_{k,e2}^* - \mathbf{p}_{k,min}^*, &  ||\mathbf{p}_{k}-\mathbf{p}_d|| \geq r_c\\
						\mathbf{p}_{k,e2}^* - \mathbf{p}_d, & ||\mathbf{p}_{k}-\mathbf{p}_d|| < r_c
						\end{array}
						\right. \\
	\label{EqPljS4Veh}
	\mathbf{p}_{k,s4} &= \mathbf{p}_{k,e2}^* - \mathbf{p}_d
\end{align}
}}

As these parameters are updated, they are used immediately by the trajectory generation algorithm. If the wind estimation and re-calculation of $r_c$ and $v_c$ have very fast update rates, the behavior of the vehicle could become erratic. To combat this, a delay could be implemented to keep $r_c$ from decreasing too soon after an increase.

\subsection{Drift Frame Velocity}
\label{SubSecDriftFrameVel}
The two modes of operation are the normal and drift modes. The difference between the two modes is whether the trajectory is generated in the inertial (normal mode) or drift (drift mode) frame. To simplify the algorithm development, all definitions from Sec.~\ref{SecTrajGen} are for the drift frame, and the normal mode operation is the special case where $||\mathbf{v}_{drift}|| = 0$. The trajectory generation in the drift frame guarantees that the vehicle achieves Property \ref{PropDriftCh6}, which comes by relaxing the tracking performance in the inertial frame.

To solve for the drift frame velocity, we refer to Theorem \ref{ThDriftVairMax} below. This theorem establishes the drift frame velocity, $\mathbf{v}_{drift}$, and the parameters to solve for the vehicle cruise velocity in the drift frame, $v_c^D$ (see Theorem 2 of \cite{Cole2018}). As long as a valid $v_c^D$ can be found, then the vehicle maintains some control authority. If no valid $v_c^D$ is found under the conditions of Theorem \ref{ThDriftVairMax} and none of the hardware/environment parameters (ex. sensor range, obstacle spacing, etc) can be modified to produce a valid cruise velocity, then the vehicle should be grounded. Otherwise, the vehicle would be carried by the wind with no ability to maneuver. 

It should be noted that Theorem \ref{ThDriftVairMax} provides inputs to Theorem 2 of \cite{Cole2018} for the maximum obstacle speed, $v_{o,max}^D$, and maximum air speed, $v_{air,max}^D$. Theorem 2 of \cite{Cole2018} is then used to solve for $v_c^D$ while respecting the constraints on vehicle thrust, sensor limitations, and the environment.

Figure \ref{FigTrajGenSummaryWithDrift} illustrates Theorem \ref{ThDriftVairMax} comparing the drift and inertial views of the environment.

\begin{figure*}
	\begin{center}
		\includegraphics[width=0.9\textwidth]{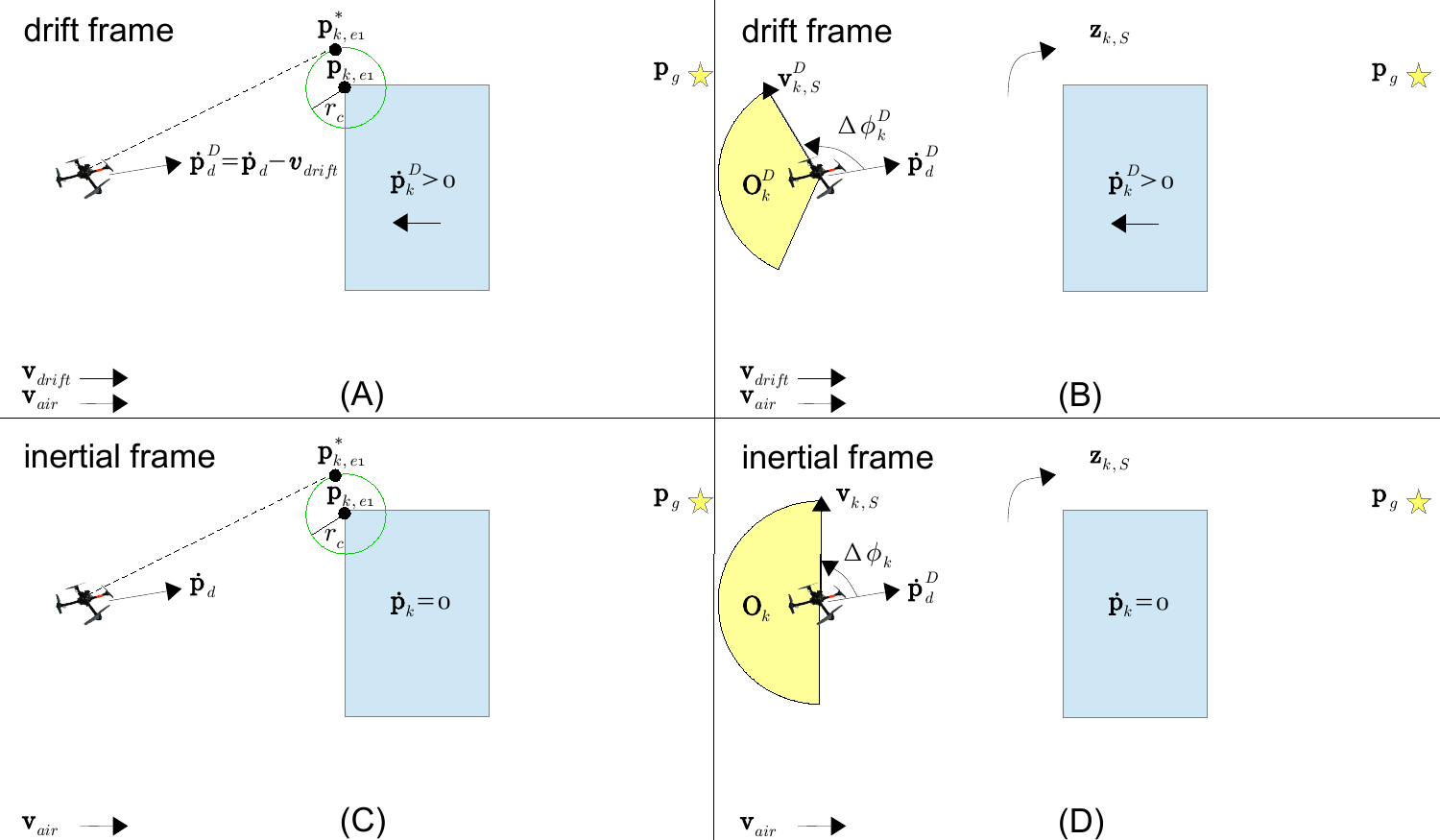}
		\caption{\label{FigTrajGenSummaryWithDrift}Summary of trajectory generation decisions comparing the drift frame to the inertial frame for navigating around a stationary obstacle. (A) In the drift frame the vehicle sees the obstacle moving towards the vehicle with non-zero velocity. The vehicle identifies the constraining projected extent point, $\mathbf{p}_{k,e1}^*$, for determining the heading change and circumnavigation direction. (B) The vehicle determines a heading change, $\Delta \phi_k$, to maneuver around the obstacle and feasible set of all possible course changes, $\mathbf{O}_k$. (C) In the inertial frame the vehicle sees the obstacle with zero velocity, but finds the same constraining projected extent points. (D) Using the projected extent point the vehicle determines the heading change and feasible set of angles.}
	\end{center}
\end{figure*}

\begin{theorem}
\label{ThDriftVairMax}
Let the drift velocity of the vehicle be constrained by
\begin{align}
	\label{EqVdrift}
	 v_{air,max} - \sqrt{\frac{f_{planar}}{K_d}} \leq &v_{drift}  \leq \sqrt{v_{air,max}^2 + \frac{f_{planar}}{K_d}}\\
	\label{EqVdriftVec}
	v_{air,max} - \sqrt{\frac{f_{planar}}{K_d}} \leq &\mathbf{v}_{drift} \cdot \frac{\mathbf{v}_{air}}{||\mathbf{v}_{air}||} \leq v_{air,max} 
\end{align}
\noindent where
\begin{align}
	\label{EqFplanarNormalDrag}
	\mathbf{f}_{planar}& = m\mathbf{\ddot{p}}_d + \mathbf{f}_w \\
	\label{EqFdrag}
	\mathbf{f}_{w} &= K_d ||\mathbf{v}_{w}||^2 (-\mathbf{x}_W) \\
	\label{EqKdDefinition}
	K_d &= \frac{1}{2} \rho C_{D} A_{x_W}
\end{align}
\noindent and $m$ is the vehicle mass, $\mathbf{v}_w = \mathbf{\dot{p}} - \mathbf{v}_{air}$ is the relative wind velocity between the vehicle and the air, $\mathbf{x}_W$ is the wind frame axis aligned with $\mathbf{v}_w$, $\rho$ is the air density, $C_{D}$ is the coefficient of drag, $A_{x_W}$ is the cross sectional area normal to $\mathbf{v}_w$.

Next, define a drift frame that moves with the constant velocity vector, $\mathbf{v}_{drift}$, relative to the inertial frame. Then in the drift frame the maximum relative air velocity is
\begin{equation}
	\label{EqVairMaxCorollary}
	\mathbf{v}_{air,max}^D = \mathbf{v}_{air} - \mathbf{v}_{drift}
\end{equation}
\noindent and the nominal vehicle velocity in the drift frame is $||\mathbf{\dot{p}}_d^D|| = 0$. Similarly, the maximum obstacle speed in the drift frame is 
\begin{equation}
	v_{o,max}^D \leq v_{drift}
\end{equation}
If $\mathbf{v}_{drift}$ is chosen to satisfy Eqs.~\ref{EqVdrift} and \ref{EqVdriftVec}, and the drift frame maximum air speed, $v_{air,max}^D$ and maximum obstacle speed $v_{o,max}^D$, are used to solve for $v_c^D$ according to Theorem 2 of \cite{Cole2018}, then in the drift frame the vehicle does not violate $f_{max}$, safely clears obstacles by $r_c$, and the inertial frame cruise velocity satisfies $v_c = v_c^D + ||\mathbf{v}_{drift}||$.
\end{theorem}
\begin{proof}
See Appendix A.
\end{proof} 

\subsection{Trajectory Guarantees}
Properties \ref{PropTrajGen1} and \ref{PropTrajGen2} are rigorously proven in \cite{Cole2018}, guaranteeing that the vehicle can navigate safely in the environment while respecting the vehicle sensor and thrust limitations and clearing all vehicles and obstacles by $r_c$. Additionally Theorem 3 in \cite{Cole2018} guarantees the vehicle reaches the goal position in finite time.

Property \ref{PropDriftCh6} guarantees that Properties \ref{PropTrajGen1} and \ref{PropTrajGen2} are maintained in both the inertial and drift frames under the assumptions from Sec.~\ref{SecAlgAssumpFormation}. The properties are upheld since the vehicle preserves control authority within the drift frame. Appendix A of \cite{ColeARXIV2019} proves that the vehicle maintains control authority provided the drift frame velocity is set using Theorem \ref{ThDriftVairMax}. 

These properties prioritize vehicle safety within the drift frame to ensure the vehicle avoids collisions; as evidenced by addressing temporary clearance radii violations immediately. Nevertheless, reaching the goal position is also important. To preserve the guarantee that the vehicle reaches the goal position in finite time we make the following additional assumption:
\begin{itemize}
	\item The wind disturbance or gust that requires the vehicle to utilize the drift frame is finite in duration; therefore, the vehicle is not drifting for an infinite amount of time.
\end{itemize}

We do not consider battery life or other vehicle specific parameters explicitly in this article but group those with Assumption \ref{AssumpGoalPos} that the vehicle is capable of reaching the goal position upon starting the mission. 

\section{Vehicle and Controller}
\label{SecVehController}
{\color{black}The vehicle dynamics for a quadrotor are given in Eqs.~\ref{EqEOMLin} and \ref{EqEOMAng}. Equation \ref{EqEOMLin} is written in the inertial frame, and Eq.~\ref{EqEOMAng} is written in the body frame: 
\begin{align}
	\label{EqEOMLin}
	m\mathbf{\ddot{p}} &= \mathbf{f} + m\mathbf{g} + \mathbf{d}_{p} \\
	\label{EqEOMAng}
	\mathbf{J}\mathbf{\dot{\omega}} &= \mathbf{\omega} \times \mathbf{J} \mathbf{\omega} + \mathbf{u} + \mathbf{R}_{IB}\mathbf{d}_{\omega}
\end{align}
\noindent where $\mathbf{f}$ is the total thrust, $\mathbf{d}_p$ is the translational disturbance (including drag), $\mathbf{J}$ is the vehicle moment of inertia, $\mathbf{\dot{\omega}}$ is the rotational acceleration, $\mathbf{u}$ is the total torque, $\mathbf{R}_{IB}$ is the rotation matrix from the inertial to body frame, and $\mathbf{d}_{\omega}$ is the rotational disturbance. The control inputs are the vehicle force, $\mathbf{f}$, and torque, $\mathbf{u}$. }

The vehicle dynamics also use the single state dynamic inflow model \cite{Peters1988,Pitt1981} to take into account aerodynamic effects on the propellers like thrust reduction from propeller inflow velocity \cite{Leishman2006,Bramwell,Newman1994,Padfield1996} and blade flapping \cite{Hoffmann}. For the purposes of the control law, these terms are added to the disturbance term.

{\color{black}The vehicle controller uses an inner-, and outer-loop control similar to \cite{Bialy2013} and \cite{Cao2016}, where the outer loop controls the translational component and the inner loop controls the rotational component. The outer loop uses a nonlinear robust integral of the sign of the error (RISE) controller \cite{FischerNL2014}, summarized in Eqs.~\ref{EqFischerControl} to \ref{EqFischere2}. The inner loop utilizes the PID control in Eq.~\ref{EqPIDControl} \cite{Cao2016}:
\begin{align}
	\label{EqFischerControl}
	\mathbf{f} &= (k_s + 1) \mathbf{e}_2 - (k_s + 1) \mathbf{e}_2(0) + \nu \\
	\mathbf{\dot{\nu}} &= (k_s + 1) \alpha_2 \mathbf{e}_2 + \beta \mathrm{sgn}(\mathbf{e}_2) \\
		\label{EqFischere2}
	\mathbf{e}_2 &= (\mathbf{\dot{p}}_d - \mathbf{\dot{p}}) + \alpha_1 (\mathbf{p}_d - \mathbf{p}) \\
	\label{EqPIDControl}
	\mathbf{u} &= k_p \mathbf{q}_d + k_i \int \mathbf{q}_d dt + k_d \mathbf{\dot{q}_d}
\end{align}

\noindent where $k_s >0$ and $\alpha_2 > 1/2$ are control gains for the translational controller and $k_p ,k_i,k_d > 0$ are the PID controller gains for the desired Euler angles, $\mathbf{q}_d$, where $\mathbf{q}_d$ are determined from $\mathbf{f}$.}

The RISE controller provides semi-global asymptotic stability under the following conditions \cite{FischerNL2014}:
\begin{property} The disturbance term and its first two time derivatives are bounded. \end{property}
\begin{property} The desired trajectory $\mathbf{p}_d \in \mathbb{R}^n$ is designed such that $\mathbf{p}_d^{(i)} \in \mathbb{R}^n, \forall i=0,\ldots,4$ exist and are bounded.\end{property}

When the RISE controller is combined with the PID inner loop, the Lyapunov stability analysis shows that the controller still achieves asymptotic trajectory tracking in the presence of disturbances. 

While these controllers are well-suited for the environments considered in this article, the trajectory generation algorithm and drift frame implementation are independent of controller. Any suitable vehicle controller can be utilized with the trajectory generation algorithm.  

\section{Simulation Scenarios}
\label{SecSimScenarioFormation}
To demonstrate the drift mode capabilities, two scenarios are examined in simulation: (A) comparison of the difference in response to two vehicles facing an extreme gust where one is utilizing the drift frame definition and the other is not, and (B) four vehicles navigating through narrow passages in different wind conditions, showing the adaptive $v_c$ and $r_c$ calculations.

\subsection{Wind Model}
The wind model used to test the algorithm performance is a realistic spatio-temporal wind field that includes turbulence and gusting \cite{ColeWindTechReport2019}. This model reflects the unique but correlated conditions that different vehicles in the operational environment experience. It utilizes the Von K\'arm\'an power spectral density (PSD) function \cite{MILF8785C1980} combined with a spreading function to obtain the 2D wind field. The gusting is based on wind farm data relationships for duration, propagation, and shape, among other parameters \cite{Branlard2009}. 

\subsection{Simulation A} 
The performance of the drift controller is most clearly demonstrated by comparing two vehicles where one enables drift mode and the other does not. Both vehicles experience the wind condition shown in Fig.~\ref{FigUnstableStableVair} where there is an extreme gust that exceeds the vehicles' thrust limitations. Both vehicles drift with the gust, but the vehicle that enables the drift mode successfully recovers.

\begin{figure}
	\begin{center}
		\includegraphics[width=3in]{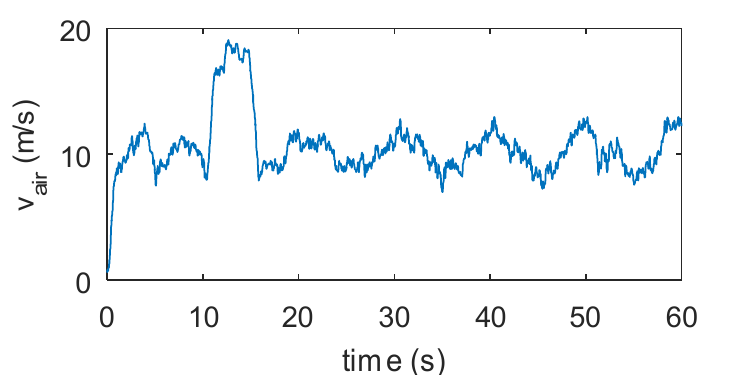}
		\caption{\label{FigUnstableStableVair} Wind experienced by both vehicles in Simulation A, where there is a gust of $v_{gust} = 19$ m/s that exceeds the vehicles' thrust limitations.} 
	\end{center}
\end{figure}
 
The vehicle parameters for the simulation are $m = 0.54$g, $f_{max}=15$N, $r_c(0) = 2$m, $r_{obs}=7$m, $r_s=12.5$m, $r_{min} = 0.09$m, $\Delta T_s = 1$s, $\Delta T_c = 0.1$s, $\mathbf{J} = \mathrm{diag}([0.0037,0.0037,0.007])$ kg/m$^2$, $C_d$ = 0.41, and $\mathbf{A}= [0.04,~0.04,~0.09]$m$^{2}$. The starting cruise velocity is computed using Theorem 2 of \cite{Cole2018} to be $v_c(0) = 1.25$m/s. The controller gains are $\alpha_1 = 0.1$, $\alpha_2 = 1$, $k_s = 0.05$, and $\beta = 0.25$, where all of these values satisfy the constraints outlined in \cite{FischerNL2014} except $\beta$, which produces non-smooth behavior for large values.

Figure \ref{FigUnstableExample}A-C shows the performance of the vehicle that does not enable drift mode. The vehicle cannot track the desired trajectory in the inertial frame and as a result the controller demands increased thrust. As the vehicle attempts to track the trajectory it sacrifices altitude and eventually crashes. Figure \ref{FigUnstableExample}D-F shows the performance of the vehicle that enables drift mode. The vehicle re-plans the trajectory in the drift frame where the trajectory respects the vehicle thrust and sensor limitations and the vehicle maintains stability to recover from the gust. 

\begin{figure*}
	\begin{center}
	\includegraphics[width=0.9\textwidth]{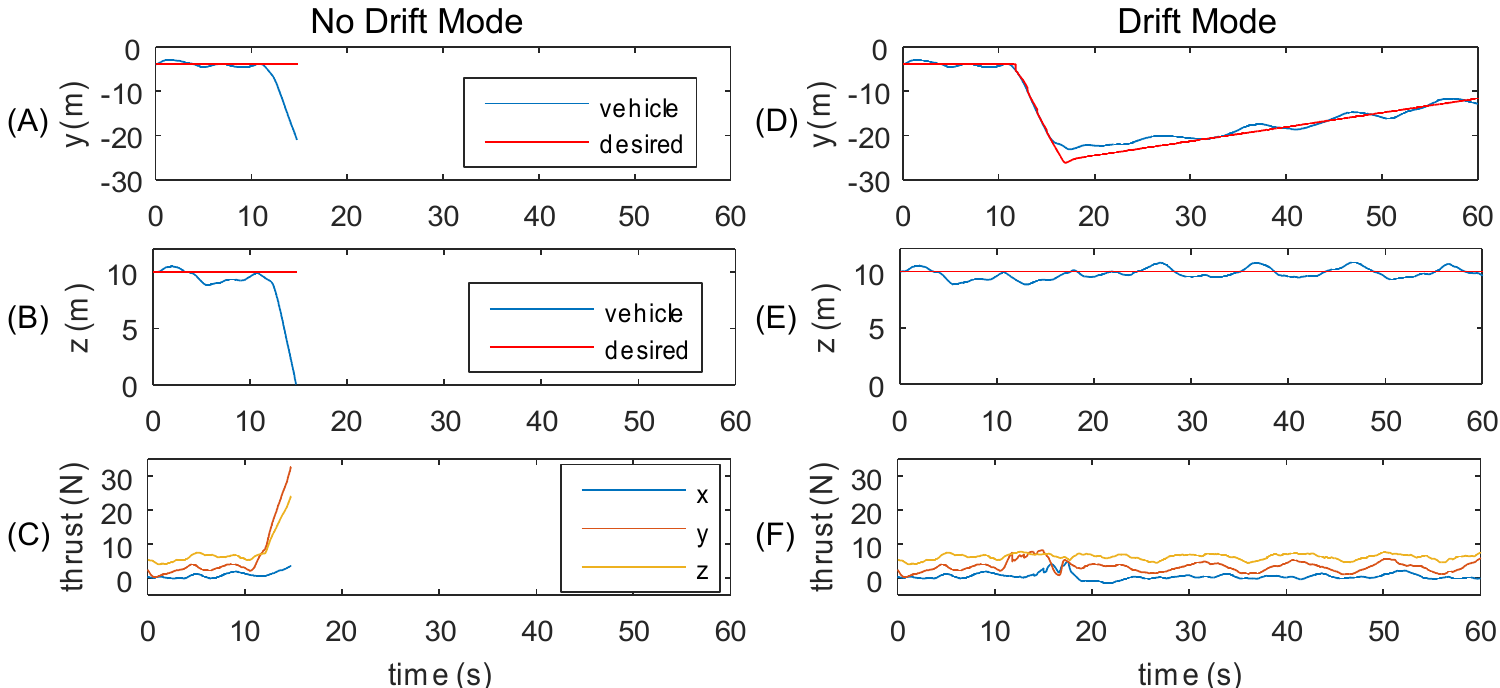}
	\caption[Example of two vehicles that experience a large gust, $v_{gust} = 19$ m/s]{\label{FigUnstableExample} Simulation A results where two vehicles experience a large gust, $v_{gust} = 18$ m/s, in the $\mathbf{y}_I$ direction. One vehicle becomes unstable as it does not enable drift mode, and the other vehicle maintains its stability by enabling drift mode. (A) The inertial frame trajectory tracking error in the $\mathbf{y}_I$ direction continues to grow as the disturbance drag force overcomes the vehicle thrust. (B) The vehicle loses altitude as it attempts to bring the tracking errors to zero, eventually resulting in a crash. (C) Since the vehicle does not have control authority in the inertial frame, the controller continues to demand a greater control force than what the vehicle can provide. (D) The drift frame trajectory error is minimized by relaxing the trajectory in the inertial frame. (E) The vehicle maintains altitude because the trajectory is feasible. (E) The desired force is within the vehicle's thrust capabilities in the drift frame.}
	\end{center}
\end{figure*}

\subsection{Simulation B}
This simulation demonstrates the dependence of the clearance radius, $r_c$, and cruise velocity $v_c$, on changing wind conditions. Four vehicles start inside a protected area where $v_{air} = 0$ m/s and move to an unprotected area where $v_{air} = 9$ m/s as shown in Fig.~\ref{FigSimAVair}. 

The vehicle parameters are the same as Simulation A except $r_c(0) = 0.8$m, $r_{obs}=5.4$m, $\alpha_1 = 0.3$, and $k_s = 0.5$. The starting cruise velocity is computed as $v_c(0) = 1.04$m/s.

Figure \ref{FigSimAXY} shows the vehicle trajectories through the environment. In the protected environment the vehicles move side by side through a narrow opening between buildings. Once the vehicles reach the unprotected environment, they recompute $v_c$ and $r_c$ based on the estimated wind. The new values for $v_c$ and $r_c$ are shown in Fig.~\ref{FigSimAVariPdotRc}A-B, respectively for vehicle 1 as a representative case. Due to the change in $r_c$, the vehicles cannot go through the next narrow opening two-by-two. Instead, the vehicles must alternate one at a time to fly through safely.

Notice that the trajectory generation algorithm respects the changing $r_c$ values as expected through the second opening in Fig.~\ref{FigSimAVariPdotRc}B. Also, it is interesting to note that since $v_c$ and $r_c$ are inversely proportional, the change in minimum reaction distance may not be significant as the two parameters may offset each other. 

\begin{figure}
	\begin{center}
		\includegraphics[width=3in]{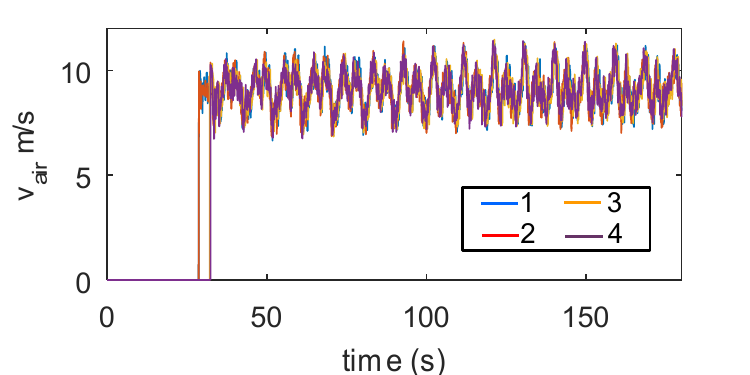}
		\caption{\label{FigSimAVair} Wind experienced by the vehicles in Simulation B. The vehicles are initially protected before being exposed to the wind.}
	\end{center}
\end{figure} 

\begin{figure}
	\begin{center}
		\includegraphics[width=3in]{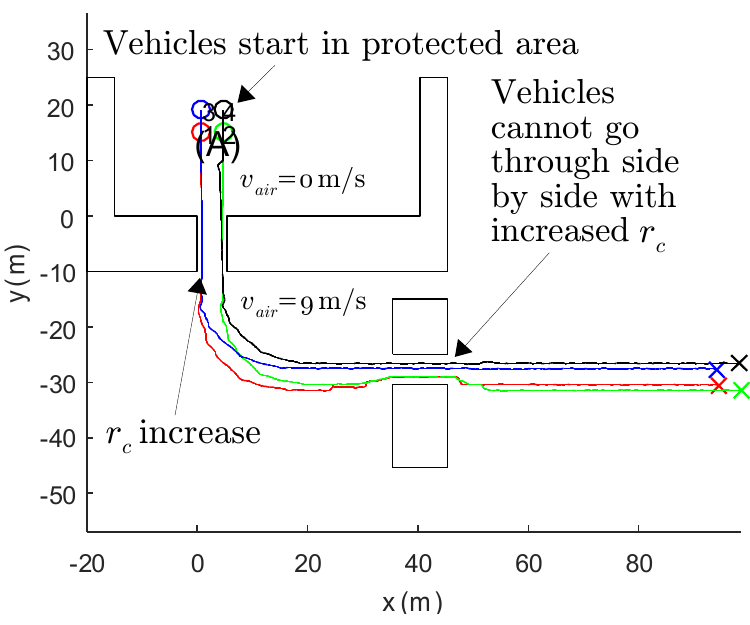}
		\caption{\label{FigSimAXY} Simulation B results showing the vehicle trajectories in the environment. The vehicles traverse through two narrow openings of the same size, once in protected conditions and once in the wind field. The trajectory generation algorithm respects the changes in $r_c$ to safely navigate the vehicles in the environment.}
	\end{center}
\end{figure}

\begin{figure}
	\begin{center}
		\includegraphics[width=3in]{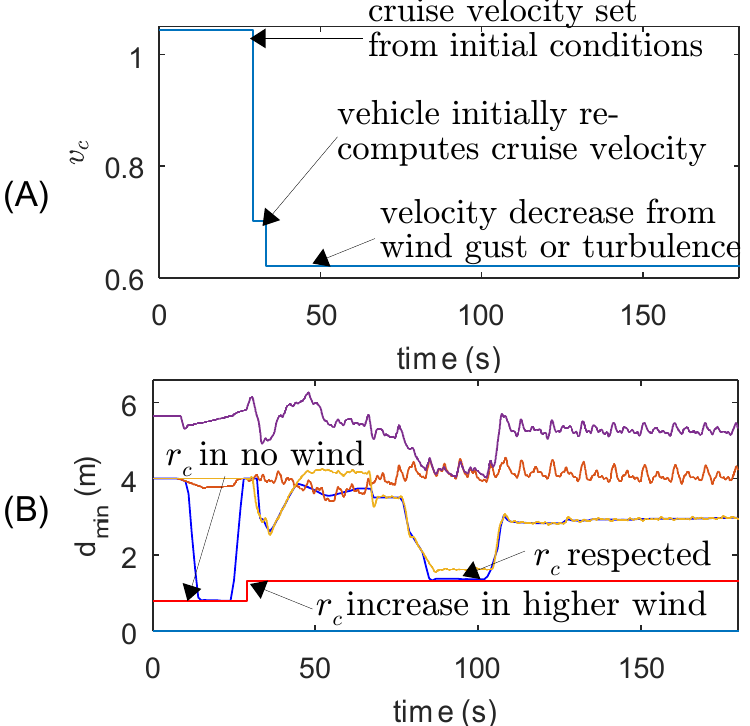}
		\caption{\label{FigSimAVariPdotRc} Variation in cruise velocity and clearance radius as a function of wind estimation for Simulation B. (A) Resulting changes in $v_c$ for the constraining (in this case weakest) vehicle. (B) Resulting change in $r_c$ values for the increased wind. Once the $r_c$ value is increased, the trajectory generation algorithm respects the new value. }
	\end{center}
\end{figure}

\section{Conclusion}
\label{SecConclusionFormation}
This article presents a trajectory generation algorithm that extends the vehicle's capable flight range beyond the closed-loop stability bounds. To do this, the algorithm establishes a drift frame in which the vehicle has some control authority, where the drift rate is set based on wind estimation measurements, vehicle thrust, and drag.

Compared to existing techniques, this trajectory planner allows the vehicles to safely navigate in environments where traditional control methods are overwhelmed. Even if the flight controllers of the traditional methods do not become unstable, the vehicle performance is unpredictable and the vehicle has no guarantees for navigation or collision avoidance. By defining the drift frame, the vehicle has some control authority that is quantitatively established, so the vehicle can determine under what conditions it can safely avoid collisions. Furthermore, if and when the disturbance decreases, the vehicle gracefully transitions back to the normal mode and re-plans a trajectory to converge to the desired inertial reference trajectory.

While the drift frame extends the environments in which the vehicle can operate, it does so by relaxing the trajectory in the inertial frame and re-planning it in the drift frame. Therefore, there is still some tracking error in an absolute sense, even though the local drift frame tracking errors are within the region of attraction. Improvements to the error in the inertial frame could be gained by optimizing the drift frame to maintain some motion towards the inertial frame trajectory to minimize errors once the disturbance decreases. Additionally, the framework developed for the drift frame is defined generically enough to extend to 3D motion. 


\begin{appendices}
\section{Proof of Theorem \ref{ThDriftVairMax}}
\label{AppendixTheorem1}
Theorem \ref{ThDriftVairMax} extends the environments in which the vehicle can operate by defining a drift frame that moves with the vehicle drift velocity and in which the vehicle has control authority. 

\begin{proof}
The vehicle navigates in the environment by computing smooth heading and velocity functions. The timespan, $\tau_{f,min}$, for these curves is defined by Theorem 1 of \cite{Cole2018}, where to navigate, $\tau_{f,min}>0 \implies a_{max} > 0$. The mean wind speed condition for $a_{max} > 0$ is defined in Theorem 1 of \cite{Cole2018} as $v_{air,max} < \sqrt{f_{planar}/K_d}$. When $v_{air,max} \geq \sqrt{f_{planar}/K_d}$, the vehicle has no control authority in the inertial frame. 

Since $v_{air,max}$ is the magnitude of the maximum resultant wind velocity experienced by the vehicle,
\begin{equation}
	v_{air,max} = ||-\mathbf{\dot{p}}_d + \mathbf{v}_{air}||
\end{equation} 
\noindent it can be reduced as $\mathbf{p}_d \rightarrow \mathbf{v}_{air}$. Define a relative frame, the drift frame, that moves according to $\mathbf{v}_{drift}$. If $\mathbf{v}_{drift}$ is chosen based on the constraints of Eq.~\ref{EqVdrift} and \ref{EqVdriftVec}, then it can be shown that in the drift frame the vehicle retains some control authority, as follows.

In the drift frame, the drift velocity, $\mathbf{v}_{drift}$, has a ``downwind" component that is aligned with $\mathbf{v}_{air}$ and a ``crosswind" component that is perpendicular to $\mathbf{v}_{air}$. There are two extreme cases, namely when $\mathbf{v}_{drift}$ is parallel to $\mathbf{v}_{air}$ and there is no crosswind component, and when the downwind component of $\mathbf{v}_{drift}$ matches $\mathbf{v}_{air}$ and there is a non-zero crosswind component. 

First consider the case where $\mathbf{v}_{drift}$ is parallel to $\mathbf{v}_{air}$ where
\begin{equation}
	\mathbf{v}_{air}^D = \mathbf{v}_{air} - \mathbf{v}_{drift} 
\end{equation}
\noindent simplifies to
\begin{equation}
	v_{air,max}^D = v_{air,max} - v_{drift}
\end{equation}

To retain control authority, $v_{air,max}^D \leq \sqrt{f_{planar}/K_d}$ from Theorem 1 of \cite{Cole2018}. To achieve this, $v_{drift}$ must be set according to
\begin{equation}
	v_{air,max}^D = v_{air,max} - v_{drift} \leq \sqrt{f_{planar}/K_d} \nonumber 
\end{equation}
\noindent which results in
\begin{equation}
	\label{EqVdriftLower}
	v_{drift} \geq v_{air,max} - \sqrt{f_{planar}/K_d}
\end{equation}
\noindent Equation \ref{EqVdriftLower} provides the lower bound of the constraint in Eq.~\ref{EqVdrift}.

In the other extreme, when the downwind component is matched to $v_{air,max}$ the resultant wind vector simplifies as
\begin{align}
	\mathbf{v}_{air}^D &= v_{air,max} \frac{\mathbf{v}_{air}}{||\mathbf{v}_{air}||} - \left(v_{air,max} \frac{\mathbf{v}_{air}}{||\mathbf{v}_{air}||} + v_{drift,\perp}\mathbf{R}_{90} \frac{\mathbf{v}_{air}}{||\mathbf{v}_{air}||}\right) \nonumber \\
	\implies \mathbf{v}_{air}^D &= v_{drift,\perp}\mathbf{R}_{90} \frac{\mathbf{v}_{air}}{||\mathbf{v}_{air}||} \nonumber \\
	\implies v_{air,max}^D & = v_{drift,\perp}
\end{align}
\noindent where $\mathbf{R}_{90}$ is the rotation matrix for the vector perpendicular to $\mathbf{v}_{air}$.

The magnitude of the crosswind component is re-written as
\begin{equation}
	v_{drift,\perp} = \sqrt{v_{drift}^2 - v_{air,max}^2}
\end{equation}

Again, to retain control authority, $v_{air,max}^D \leq \sqrt{f_{planar}/K_d}$, which gives the following:
\begin{align}
	v_{air,max}^D = v_{drift,\perp} \leq \sqrt{f_{planar}/K_d} \nonumber \\
	 \implies \sqrt{v_{drift}^2 - v_{air,max}^2} \leq \sqrt{f_{planar}/K_d} \nonumber \\
	 \implies v_{drift}^2 \leq v_{air,max}^2 + \frac{f_{planar}}{K_d} \nonumber \\
	 \label{EqVdriftUpper}
	\implies  v_{drift} \leq \sqrt{v_{air,max}^2 + \frac{f_{planar}}{K_d} }
\end{align}
\noindent where Eq.~\ref{EqVdriftUpper} provides the upper bound of the constraint in Eq.~\ref{EqVdrift}.

Equations \ref{EqVdriftLower} and \ref{EqVdriftUpper} constrain the magnitude, but the drift frame velocity vector must also be constrained. Considering the extreme cases again, to satisfy the downwind component the following must hold
\begin{equation}
	\label{EqVdriftVecUpper}
	\mathbf{v}_{drift} \cdot \frac{\mathbf{v}_{air}}{||\mathbf{v}_{air} ||} \geq v_{air,max} - \sqrt{f_{planar}/K_d}
\end{equation}

Similarly, to enforce the upper bound on the velocity magnitude and avoid unnecessarily high velocities, the drift frame downwind component must not exceed $v_{air,max}$, giving
\begin{equation}
	\label{EqVdriftVecLower}
	\mathbf{v}_{drift} \cdot \frac{\mathbf{v}_{air}}{||\mathbf{v}_{air} ||} \leq v_{air,max}
\end{equation}

Equations \ref{EqVdriftVecUpper} and \ref{EqVdriftVecLower} provide the bounds for Eq.~\ref{EqVdriftVec}. 

Lastly, the obstacle velocities must also be translated to the drift frame. According to Assumptions \ref{AssumpLessCapable} and \ref{AssumpMovingObsFormation}, at best a moving obstacle moves at the same velocity as the vehicles, but it never moves against the wind towards a vehicle. Therefore, the constraining case for the obstacles is a stationary obstacle, which in the drift frame moves with speed
\begin{equation}
	v_{o,max}^D = v_{drift}
\end{equation}
By using the drift frame parameters, $v_{air,max}^D$, and $v_{o,max}^D$, where it is shown that the vehicle has control authority, a drift frame cruise velocity $v_c^D$ is computed using the conditions of Theorem 2 of \cite{Cole2018}, which guarantees the vehicle does not violate $f_{max}$ and safely clears obstacles by $r_c$.   
\end{proof}
\end{appendices}

\bibliographystyle{unsrt}
\bibliography{C:/Users/cole/Documents/ColeThesisWork/Documentation/ReferencesBibFiles/Controls,C:/Users/cole/Documents/ColeThesisWork/Documentation/ReferencesBibFiles/EstimationFormation,C:/Users/cole/Documents/ColeThesisWork/Documentation/ReferencesBibFiles/Gusting,C:/Users/cole/Documents/ColeThesisWork/Documentation/ReferencesBibFiles/RCSBib,C:/Users/cole/Documents/ColeThesisWork/Documentation/ReferencesBibFiles/WindWaterRefs}

\vspace{-0.5in}
\begin{IEEEbiography}[{\includegraphics[width=1in]{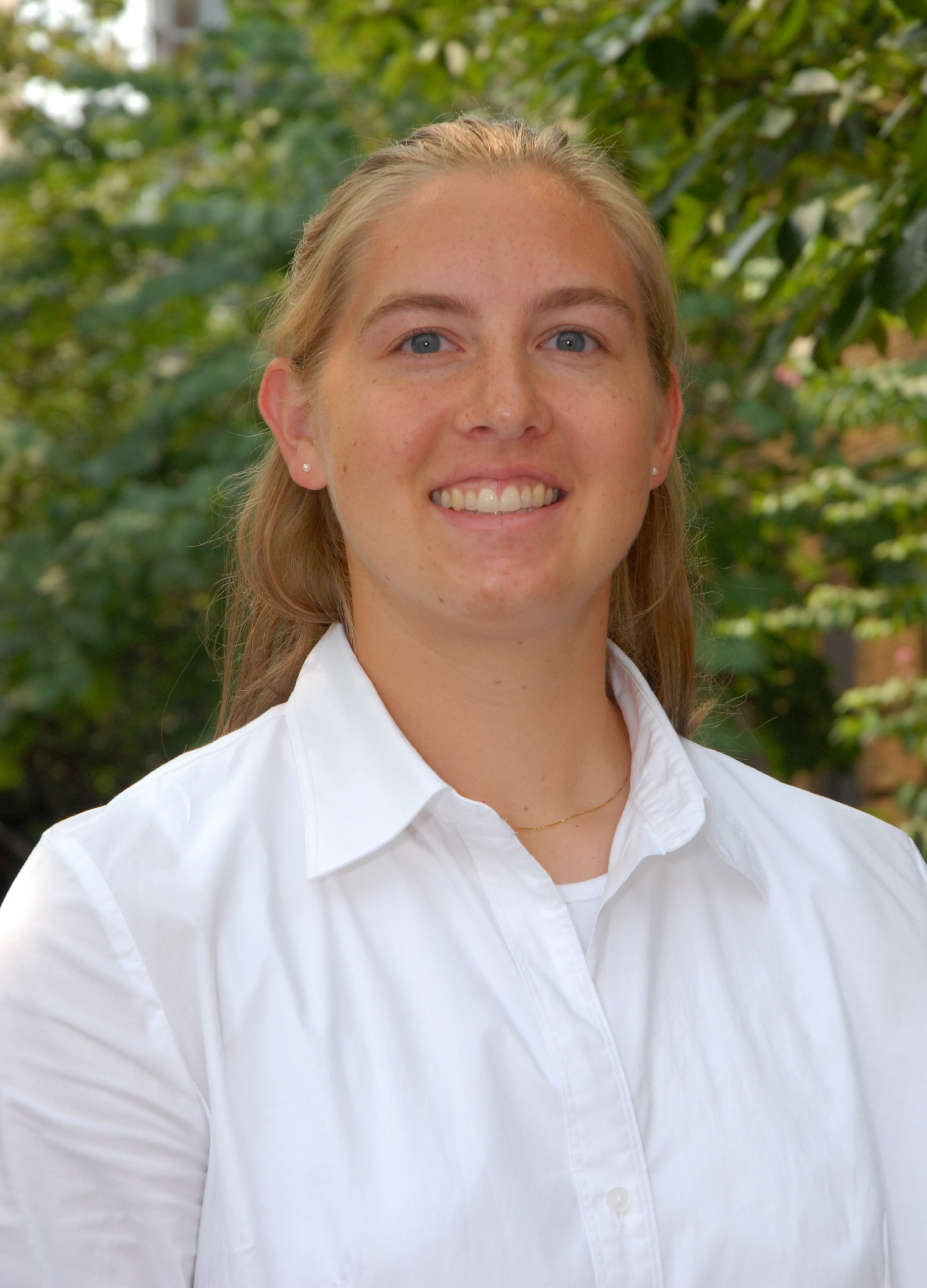}}]{Kenan Cole}
Kenan Cole received her B.S., M.S., and Ph.D. degrees in mechanical engineering from The George Washington University, Washington, DC, USA, in 2007, 2011, and 2018, respectively. Her research interests include vehicle controls, formation controls, and modeling environmental disturbances. 
\end{IEEEbiography}
\vspace{-5in}
\begin{IEEEbiography}[{\includegraphics[width=1in]{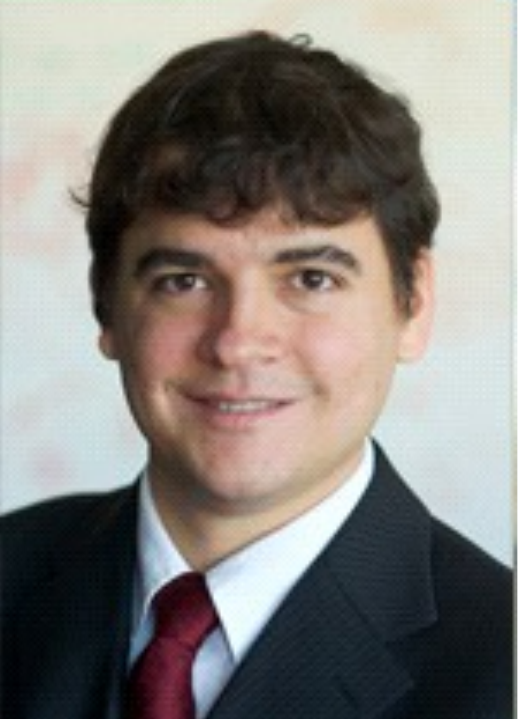}}]{Adam Wickenheiser}
Adam M. Wickenheiser received the B.S. degree in mechanical engineering (with a minor in applied mathematics) and M.S. and Ph.D. degrees in aerospace engineering from Cornell University, Ithaca, NY, USA, in 2002, 2006, and 2008, respectively. Since 2018, he has been a faculty member with the Department of Mechanical Engineering, University of Delaware, Newark, DE, USA, where he is currently an Associate Professor.

From 2010 to 2018, he was an Assistant Professor with the Department of Mechanical and Aerospace Engineering, The George Washington University. From 2008 to 2009, he was a Postdoctoral Associate with the Sibley School of Mechanical and Aerospace Engineering, Cornell University. His research interests include bioinspired flight, multifunctional materials and systems, and energy harvesting for autonomous systems.

Prof. Wickenheiser was the Chair of the Energy Harvesting Technical Committee of the American Society of Mechanical Engineers from 2014 to 2016, and is currently a member of the International Organizing Committee for the International Conference on Adaptive Structures and Technologies. He was the recipient of the 2011 Intelligence Community Young Investigator Award.
\end{IEEEbiography}

\end{document}